\documentclass[11pt]{article}
\usepackage{amssymb}
\usepackage{amsmath}
\usepackage{amsthm}
\usepackage{mathtools}
\usepackage[textwidth=16cm, textheight=24cm, marginratio=1:1]{geometry}

\usepackage[english]{babel}
\usepackage[utf8]{inputenc}
\usepackage[T1]{fontenc}
\usepackage{todonotes}
\usepackage{xcolor}
\usepackage{xfrac}

\usepackage[pdfborder={0 0 0}]{hyperref}
\hypersetup{
colorlinks,
linkcolor={red!80!black},
citecolor={blue!80!black},
urlcolor={blue!80!black}
}
\usepackage{setspace}
\usepackage{biblatex}
\usepackage{csquotes}

\newtheorem{theorem}{Theorem}
\newtheorem{corollary}[theorem]{Corollary}
\newtheorem{lemma}[theorem]{Lemma}

\newtheorem{question}[theorem]{Question}
\newtheorem{proposition}[theorem]{Proposition}
\theoremstyle{definition}

\newtheorem{remark}[theorem]{Remark}
\newtheorem{example}[theorem]{Example}

\DeclareMathOperator{\im}{im}
\DeclareMathOperator{\Spec}{Spec}

\title{Irreversibility of Structure Tensors of Modules}
\author{Maciej Wojtala}
\date{\today{}}

\begin{document}

\maketitle

\begin{abstract}
    Determining the matrix multiplication exponent $\omega$ is one of the greatest open problems in theoretical computer science. We show that it is impossible to prove $\omega = 2$ by starting with structure tensors of modules of fixed degree and using arbitrary restrictions. It implies that the same is impossible by starting with $1_A$-generic non-diagonal tensors of fixed size with minimal border rank. This generalizes the work of Bl\"aser and Lysikov \cite{blser_et_al:LIPIcs:2020:12686}. Our methods come from both commutative algebra and complexity theory.
\end{abstract}


\section{Introduction}
\label{section-introduction}

Determining the matrix multiplication exponent $\omega$ is one of the most important problems in theoretical computer science. Trivial bounds are $2 \leq \omega \leq 3$. In the classical paper V. Strassen proved a non-trivial bound $\omega \leq \log_2 7 < 2.81$ \cite{strassen1969gaussian}. D. Coppersmith and S. Winograd proved the bound $\omega < 2.376$ \cite{CW}. This result was recently slightly improved \cite{davie_stothers_2013,Williams2012MultiplyingMF,LeGall} resulting with the best known upper bound $\omega < 2.373$ with rounding to the third decimal place.

No non-trivial lower bound is known and the conjecture states that $\omega = 2$. However since Coppersmith-Winograd there was very little progress in inventing more efficient algorithms and obtaining better upper bounds. Recent papers give an explanation for the phenomenon - many currently used approaches cannot result with an algorithm giving $\omega = 2$. For instance the laser method used for big Coppersmith-Winograd tensors cannot show $\omega = 2$, in fact it cannot even show $\omega \leq 2.30$ \cite{alman_et_al:LIPIcs:2018:8360,Alman_Williams}. Also the framework proposed by Umans and Cohn using reducing matrix multiplication to group algebra multiplication cannot show $\omega = 2$ for abelian groups and certain non-abelian groups \cite{Blasiak_2017,blasiak2017groups}. M. Christiandl, P. Vrana and J. Zuiddam introduced a quantity called irreversibility and proved that it is impossible to show $\omega = 2$ using arbitrary restrictions starting with irreversible tensors (i.e. with irreversibility greater than one) \cite{Christandl_Vrana_Zuiddam2}.

M. Bl\"aser and V. Lysikov showed in their paper \cite{blser_et_al:LIPIcs:2020:12686}, that one cannot prove $\omega = 2$ over $\mathbb{C}$ using arbitrary restrictions and starting with powers of structure tensors of non-semisimple algebras with bounded dimension.
This result is quite general, since the class of tensors that are structural tensors of some algebra is quite large, larger than previously considered classes, and using arbitrary restrictions is not a restrictive assumption.
For instance, Coppersmith-Winograd tensors are in this class. Let us name the three coordinates of our tensors by $V_1, \: V_2, \: V_3$, i.e. our tensors belong to the space $V_1 \otimes V_2 \otimes V_3$ where $V_i \simeq \mathbb{C}^n$.
Bl\"aser and Lysikov use this result to conclude that it is impossible to show $\omega = 2$ using arbitrary restrictions and starting with tensors that are both $1_{V_1}$- and $1_{V_2}$-generic (so-called binding tensors) with minimal border rank and being non-diagonal.
This result is really interesting since small border rank is believed to be desirable in obtaining fast matrix multiplication.
An important note is that one could still try to prove $\omega = 2$ by taking sequence of dimensions going to infinity (since in the result there is an assumption that dimensions are bounded).

This paper is an extension of results obtained by Bl\"aser and Lysikov. We consider an arbitrary algebraically closed field $\mathbb{K}$ (of arbitrary characteristic) and tensors from $\mathbb{K}^n \otimes \mathbb{K}^n \otimes \mathbb{K}^n$. From the perspective of commutative algebra a natural generalisation is to consider modules instead of algebras. In this paper we show such a generalisation. To achieve that we define a structure tensor of a module.
The generalisation theorem is the main result of this paper:

\begin{theorem}
\label{thm-main 1}
    For bounded $n$ it is impossible to prove $\omega = 2$ over $\mathbb{K}$ using arbitrary restrictions and starting with powers of tensors of size $n$ that are isomorphic to some structure tensor of a non-semisimple module.
\end{theorem}

This allows us also to extend the corollary obtained by Bl\"aser and Lysikov:

\begin{corollary}
\label{corollary-main-1}
     For bounded $n$ it is impossible to prove $\omega = 2$ over $\mathbb{K}$ using arbitrary restrictions and starting with powers of tensors of size $n$ that are $1_{V_1}$-generic, have minimal border rank and have rank larger than $n$.
\end{corollary}

It means that we only need to assume $1_{V_1}$-genericity - assuming $1_{V_2}$-genericity is not necessary. There are plenty of tensors that are $1_{V_1}$-generic but not $1_{V_2}$- or $1_{V_3}$-generic; our generalization applies to them as well.

The approach in the proof is based on the one in Bl\"aser and Lysikov paper \cite{blser_et_al:LIPIcs:2020:12686}, but there are several issues.
The main part focuses on showing that a structure tensor of a non-semisimple module is $0$-subtight-unstable. To achieve this we use the quotient ring $A$ obtained by dividing the polynomial ring by the annihilator of the considered module. We show that ring $A$ is Artinian. It allows to create a filtration of the module induced by powers of a nilradical of $A$. We also induce the filtration of the polynomial ring (by taking the preimage of the filtration of $A$) and of its subset consisting of forms of degree at most one (by restricting the filtration of the polynomial ring).
Then we will show that it suffices to prove that the last filtration is non-trivial. To prove that this filtration is indeed non-trivial we will analyze the spectrum of the quotient ring and by Hilbert's Nullstellensatz use it to analyze affine functions which correspond to forms of degree at most one.

It is unclear if the assumption of $1_{V_1}$-genericity can be replaced with conciseness; the following question remains open:

\begin{question}
Do there exist concise minimal border rank (with rank greater than border rank) tensors that are stable?
\end{question}

\subsection*{Acknowledgements}
{During preparation of this publication the author was part of the Szkoła Orłów programme. The publication was created under the supervision of Joachim Jelisiejew, whose help and support were invaluable. The author also wants to thank Markus Bl\"aser, Vladimir Lysikov, Joseph Landsberg and the anonymous referee for helpful comments.}

\section{Definitions}
\label{section-definitions}

Let $\mathbb{K}$ be an algebraically closed field and let $S = \mathbb{K}[x_1, \: x_2, \: \ldots, \: x_{n-1}]$.

We denote by $S_{\leq 1}$ the $\mathbb{K}$-linear subspace of polynomials of degree at most one.

Let us consider an $S$-module $M$. Observe that structure of multiplication in this module is uniquely determined by multiplication by elements of $S_{\leq 1}$. Furthermore, multiplication by $S_{\leq 1}$ is uniquely determined by the multiplication by a basis of $S_{\leq 1}$. We define the structure tensor of $M$ as $\sigma_M \in S_{\leq 1}^* \otimes M^* \otimes M$ or equivalently $\sigma_M \colon S_{\leq 1} \otimes M \rightarrow M$. For a choice of coordinates on the first, second and third factors ($\{A_i\}, \: \{B_j\}, \{C_k\}$ respectively), $\sigma_M$ simply encodes the result of multiplying $A_i$ by $B_j$ in $\{C_k\}$ basis.

Note that using the polynomial algebra is a quite general assumption.
Indeed, let $A$ be a commutative unital algebra with $\dim_{\mathbb{K}} A = n$, let $M$ be a module over $A$ of dimension $n$ and $t$ be the tensor corresponding to the bilinear map $A \otimes M \to M$.
Let $(1, \: g_1, \: g_2, \ldots, \: g_{n-1})$ span the algebra as a $\mathbb{K}$-linear space.
Consider the unique surjection $S \to A$ that sends $x_i$ to $g_i$.
Then $M$ becomes an $S$-module and $t$ identifies with the structure tensor of the $S$-module $M$.
In particular, after putting $M = A$ we deduce that the structure tensor of any commutative unital algebra with $\dim_{\mathbb{K}} A  = n$ is isomorphic to the structure tensor of the obtained $S$-module.
Note that using the unity is needed to assure $1_{V_1}$-genericity. In comparison to \cite{blser_et_al:LIPIcs:2020:12686} we assume that the algebra $A$ is commutative, however due to \cite[Lemma 2.6]{LaMM} it is satisfied when its structure tensor is of minimal border rank (so Corollary \ref{corollary-main-1} strengthens analogous result by Bl\"aser and Lysikov \cite[Corollary 25]{blser_et_al:LIPIcs:2020:12686}).

\begin{example}
\label{example-4_x_4_matrices}
We consider $\mathbb{K}^4$, where we treat elements as column vectors.
Let us also consider matrices:
\[A_1 = \begin{bmatrix}
0 & 0 & 1 & 0 \\
0 & 0 & 0 & 0 \\
0 & 0 & 0 & 0 \\
0 & 0 & 0 & 0 \\
\end{bmatrix},
\qquad 
A_2 = \begin{bmatrix}
0 & 0 & 0 & 1 \\
0 & 0 & 0 & 0 \\
0 & 0 & 0 & 0 \\
0 & 0 & 0 & 0 \\
\end{bmatrix},
\qquad
A_3 = \begin{bmatrix}
0 & 0 & 0 & 0 \\
0 & 0 & 1 & 0 \\
0 & 0 & 0 & 0 \\
0 & 0 & 0 & 0 \\
\end{bmatrix}.\]

The matrices clearly pairwise commute, so we can introduce a structure of a $\mathbb{K}[x_1, \: x_2, \: x_3]$-module on $\mathbb{K}^4$ where multiplication of a vector $v$ by $x_i$ is simply the left multiplication of $v$ by $A_i$ (so we identify $x_i$ with $A_i$ and take as a multiplication in the module the left multiplication of a vector by a matrix).
Let us take the standard basis $(e_1, \: e_2, \: e_3, \: e_4)$ of $\mathbb{K}^4$.
We also have the basis of $\mathbb{K}[x_1, \: x_2, \: x_3]_{\leq 1}$: $(A_0 \coloneqq Id, \: A_1, \: A_2, \: A_3)$ (since we identified $x_i$ with $A_i$ for $i = 1, \: 2, \: 3$).
To obtain the structure tensor we need to verify the results of pairwise multiplications.

We have

\[A_0 \cdot e_i = e_i,
\qquad 
A_1 \cdot e_i =
    \begin{bmatrix}
        \delta_{i = 3} \\
        0 \\
        0 \\
        0 \\
    \end{bmatrix},
\qquad
A_2 \cdot e_i =
    \begin{bmatrix}
        \delta_{i = 4} \\
        0 \\
        0 \\
        0 \\
    \end{bmatrix},
\qquad
A_3 \cdot e_i =
    \begin{bmatrix}
        0 \\
        \delta_{i = 3} \\
        0 \\
        0 \\
    \end{bmatrix}.\]
    
    So the structure tensor of this module is
    \[A_0^* \otimes e_1^* \otimes e_1 + A_0^* \otimes e_2^* \otimes e_2 + A_0^* \otimes e_3^* \otimes e_3 + A_0^* \otimes e_4^* \otimes e_4 + A_1^* \otimes e_3^* \otimes e_1 + A_2^* \otimes e_4^* \otimes e_1 +  A_3^* \otimes e_3^* \otimes e_2.\]
\end{example}

Above we explained that every structure tensor of a commutative unital algebra is the structure tensor of a module. Here we show that for an algebra $A = \sfrac{S}{I}$ the structure tensor of $A$ as an algebra and as an $S$-module may differ. The difference comes from the fact that a given algebra can have many $S$-module structures: the structure coming from a span $(1, \: g_1, \: g_2, \ldots, \: g_{n-1})$ is in general different from the structure coming from $\sfrac{S}{I}$.

\begin{example}
\label{example-algebra_module}
Let us consider the algebra $\sfrac{\mathbb{K}[x_1, \: x_2]}{(x_1^3, \: x_2)}$. We have standard basis of $\mathbb{K}[x_1, \: x_2]_{\leq 1}$: $(A_0 \coloneqq 1, \: A_1 \coloneqq x_1, \: A_2 \coloneqq x_2)$ and standard basis of the algebra: $(A_0, \: A_1, \: A_3 \coloneqq x_1^2)$. To obtain the structure tensors we need to verify results of pairwise multiplications, however for structure tensor of module we multiply elements from the basis of $\mathbb{K}[x_1, \: x_2]_{\leq 1}$ by the basis of the algebra and for the structure tensor of algebra - elements from the basis of the algebra by themselves.

So the structure tensor of the module and the algebra are respectively
\begin{align*}
&A_0^* \otimes A_0^* \otimes A_0 + A_0^* \otimes A_1^* \otimes A_1 + A_0^* \otimes A_3^* \otimes A_3 + A_1^* \otimes A_0^* \otimes A_1 + A_1^* \otimes A_1^* \otimes A_3,\\
&A_0^* \otimes A_0^* \otimes A_0 + A_0^* \otimes A_1^* \otimes A_1 + A_0^* \otimes A_3^* \otimes A_3 + A_1^* \otimes A_0^* \otimes A_1 + A_1^* \otimes A_1^* \otimes A_3 + A_3^* \otimes A_0^* \otimes A_3.
\end{align*}

Note that both structure tensors come from the $\mathbb{K}[x_1, \: x_2]$-module structure on $A$: the structure tensor of the module encodes the results of multiplication by $1, \: x_1, \: x_2$, while the structure tensor of the algebra encodes the results of multiplication by $1, \: x_1, \: x_1^2$.
\end{example}

A module is simple if it is non-zero and has no non-zero proper submodules. A module is semi-simple if it is a direct sum of simple modules.

An $\emph{arbitrary restriction}$ of a tensor $t \in V_1 \otimes V_2 \otimes V_3$ is $\phi(t)$, where $\phi \colon V_1 \otimes V_2 \otimes V_3 \to V_1^\prime \otimes V_2^\prime \otimes V_3^\prime$ is a linear map induced by a triple of linear maps $\phi_i \colon V_i \to V_i^\prime$.

For a tensor $t \in V_1 \otimes V_2 \otimes V_3$ and a linear form $x \in V_1^*$, the contraction $t \cdot x$ is defined
as $(v_1 \otimes v_2 \otimes v_3) \cdot x = x(v_1)(v_2 \otimes v_3)$ for rank one tensors and extended to arbitrary tensors
by linearity. Thus, a tensor $t \in V_1 \otimes V_2 \otimes V_3$ defines a map $V_1^* \rightarrow V_2 \otimes V_3$ sending $x$ to $t \cdot x$.
The two other maps $V_2^* \rightarrow V_1 \otimes V_3$ and $V_3^* \rightarrow V_1 \otimes V_2$ can be defined similarly. These maps are called flattenings of the tensor $t$. A tensor is called concise if all its flattenings are injective.
Such a tensor does not lie in any non-trivial subspace $V_1^\prime \otimes V_2^\prime \otimes V_3^\prime$ with  $V_k^\prime \subset V_k$. We denote
the maximum of the three ranks of the flattenings by $N(t)$. For a concise tensor, the ranks
of the flattenings are the dimensions of $V_k$, and $N(t) = \max \left \{\dim_{\mathbb{K}} V_1, \: \dim_{\mathbb{K}} V_2, \: \dim_{\mathbb{K}} V_3 \right \}$.
A tensor $t \in V_1 \otimes V_2 \otimes V_3$ is called $1_{V_1}$-generic if $\dim_{\mathbb{K}} V_2 = \dim_{\mathbb{K}} V_3$ and there exists $x \in V_1^*$ such that the matrix $t \cdot x \in V_2 \otimes V_3$ has full rank. The notions of $1_{V_2}$-genericity and $1_{V_3}$-genericity are defined analogously \cite{blser_et_al:LIPIcs:2020:12686}.

A block tensor is a tensor $t \in V_1 \otimes V_2 \otimes V_3$ with a triple of direct sum decompositions $V_1 = \bigoplus_{i \in I_1} V_{1, \: i}, \: V_2 = \bigoplus_{i \in I_2} V_{2, \: i}, \: V_3 = \bigoplus_{i \in I_3} V_{3, \: i}$.
The decompositions of $V_k$ induce the decomposition of the tensor space
$V_1 \otimes V_2 \otimes V_3 = \bigoplus_{(i_1, \: i_2, \: i_3) \in I_1 \times I_2 \times I_3} V_{1, \: i_1} \otimes V_{2, \: i_2} \otimes V_{3, \: i_3}$.
For a block tensor $t$, we denote by $t_{i_1 i_2 i_3}$ its projection onto $V_{1, \: i_1} \otimes V_{2, \: i_2} \otimes V_{3, \: i_3}$ \cite[Definition 8]{blser_et_al:LIPIcs:2020:12686}.
The support of a block tensor $t$ is defined as $supp \; t = \{(i_1, \: i_2, \: i_3) \in I_1 \times I_2 \times I_3 \mid t_{i_1 i_2 i_3} \neq 0\}$ \cite[Definition 9]{blser_et_al:LIPIcs:2020:12686}.

The block format of a block tensor is a triple $(n_1, \: n_2, \: n_3)$ of maps $n_k \colon I_k \rightarrow \mathbb{N}$ defined as $n_k(i) = \dim_{\mathbb{K}} V_{k, \: i}, \: k = 1, \: 2, \: 3$.  The relative block format is a triple $(f_1, \: f_2, \: f_3)$ defined as $f_k(i) = \frac{n_k(i)}{N_k}$ where $N_k = \dim_{\mathbb{K}} V_k$. \cite[Definition 10]{blser_et_al:LIPIcs:2020:12686}.
A subset $S \subset I_1 \times I_2 \times I_3$ is called s-subtight with numbering given by three maps $a_k \colon I_k \rightarrow \mathbb{Z}$ if for each $(i_1, \: i_2, \: i_3) \in S$ we have $a_1(i_1) + a_2(i_2) + a_3(i_3) \leq s$ \cite[Definition 13]{blser_et_al:LIPIcs:2020:12686}.

A tensor $t \in V_1 \otimes V_2 \otimes V_3$ is unstable, if $0$ is contained in the Zariski closure of $SL(V_1) \times SL(V_2) \times SL(V_3)$ orbit of $t$.

A set $S \in I_1 \times I_2 \times I_3$ is a combinatorially unstable support in block format $(n_1, \: n_2, \: n_3)$ if there exist exponents $u_k \colon I_k \rightarrow \mathbb{Q}$ such that $\sum_{i \in I_k} n_k(i) u_k(i) = 0$ for each $k$ and $u_1(i_1) + u_2(i_2) + u_3(i_3) > 0$ for each $(i_1, \: i_2, \: i_3) \in S$ \cite[Definition 19]{blser_et_al:LIPIcs:2020:12686}.
A tensor is combinatorially unstable if its support is a combinatorially unstable support.

It turns out that combinatorially unstable tensors are unstable \cite[Proposition 20]{blser_et_al:LIPIcs:2020:12686}.

Let a tensor $t$ have $s$-subtight support with numbering $(a_1, \: a_2, \: a_3)$. Let $(f_1, \: f_2, \: f_3)$ be a relative block format of $t$ and $\overline{a}_k = \sum_{i \in I_k} f_k(i) a_k(i)$ for $k = 1, \: 2, \: 3$. If the inequality $\overline{a}_1 + \overline{a}_2 + \overline{a}_3 > s$ holds, then we say then that the tensor $t$ is $s$-subtight-unstable.
It turns out that $s$-subtight-instability implies combinatorial instability \cite[Theorem 21]{blser_et_al:LIPIcs:2020:12686}.

By $SR(t)$ we denote slice rank of tensor $t$.
By $\widetilde{SR} (t)$ we denote the asymptotic slice rank of tensor $t$, i.e. $\widetilde{SR} (t) = \limsup_{m \in \mathbb{N}}{SR(t^{\otimes m})^{\frac{1}{m}}}$.

For a ring $R$ we denote by $rad(R)$ the nilradical of $R$, i. e. the ideal consisting of elements which raised to some power give zero.

\section{Irreversibility of Structure Tensors of Modules}
\label{section-main-part}

Let $\mathbb{K}$ be an algebraically closed field.
Let $S = \mathbb{K}[x_1, \: x_2, \: \ldots, \: x_{n-1}]$.
For $S$-module $M$ of rank $n = \dim_{\mathbb{K}} S_{\leq 1}$ let $\sigma_M \colon S_{\leq 1} \otimes M \rightarrow M$ be its structure tensor, so $\sigma_M \in S_{\leq 1}^* \otimes M^* \otimes M \simeq \mathbb{K}^{n} \otimes \mathbb{K}^{n} \otimes \mathbb{K}^{n}$.

Let $F_n$ be the set of tensors $t \in \mathbb{K}^{n} \otimes \mathbb{K}^{n} \otimes \mathbb{K}^{n}$ such that there exist a non-semisimple $S$-module $M$ such that $t \simeq \sigma_M$ (so $F_n$ is the set of tensors considered in Theorem \ref{thm-main 1}).

Let us distinguish the coordinates by taking $\mathbb{K}^{n} \otimes \mathbb{K}^{n} \otimes \mathbb{K}^{n} = A \otimes B \otimes C$.

Before proving this corollary we introduce a useful lemma which classifies simple modules. It is quite known but we present it for self-containment.

\begin{lemma}
\label{lemma-simple-module}
   An $S$-module $M$ is simple $\iff$ there exists a maximal ideal $\mathfrak{m}$ of $S$ such that $M = \sfrac{S}{\mathfrak{m}}$.
\end{lemma}
\begin{proof}
    "$\impliedby$": $\sfrac{S}{\mathfrak{m}}$ is a field which is clearly simple.
    
    "$\implies$": Let $N$ be a simple $S$-module. Let $n_0$ be a non-zero element of $N$. Then $Sn_0$ is a non-trivial submodule and thus $Sn_0 = N$. Let us define $\pi_N \colon S \rightarrow N$, $\pi_N(s) = sn_0$. By $Sn_0 = N$, we have that $\pi_N$ is surjective and thus $N \simeq \sfrac{S}{\ker \pi_N}$. It now suffices to show that $\ker \pi_N$ is a maximal ideal of $S$.
    
    Since $N \neq 0$, $\ker \pi_N$ is a subideal of some maximal ideal $\mathfrak{m}$ of $S$. Let us suppose that $\ker \pi_N \subsetneq \mathfrak{m}$.
    Then $\sfrac{\mathfrak{m}}{\ker \pi_N} \neq 0$ and clearly $\sfrac{\mathfrak{m}}{\ker \pi_N} \subseteq \sfrac{S}{\ker \pi_N} = N$. However $N$ is simple, so it must hold $\sfrac{\mathfrak{m}}{\ker \pi_N} = \sfrac{S}{\ker \pi_N}$.
    However it implies that $\sfrac{
    S}{\mathfrak{m}} = \sfrac{(\sfrac{S}{\ker \pi_N})}{\sfrac{(\mathfrak{m}}{\ker \pi_N})} = 0$, which is a contradiction.
\end{proof}

\begin{proof}[Proof of Corollary \ref{corollary-main-1}]
\label{proof-corollary-main-1}
    Let $T_n$ be the set of tensors $t \in \mathbb{K}^{n} \otimes \mathbb{K}^{n} \otimes \mathbb{K}^{n}$ such that $t$ is $1_A$-generic, the rank of $t$ is larger than $n$ and the border rank of $t$ is $n$.
    Let $t \in T_n$. We need to show that $t \in F_n$.
    
    By the assumption that $t$ is $1_A$-generic we know that there exist $\alpha$ such that $t(\alpha): C \rightarrow B^*$ has maximal rank.
    Thus $a \otimes b \otimes c \mapsto a \otimes b \otimes t(\alpha)(c)$ is a tensor isomorphism between $t$ and a tensor $\widetilde{t} \in A \otimes B \otimes B^*$.
    Let us denote $V = \widetilde{t}(A^*) \subseteq B \otimes B^* = 
    End(B)$.
    By \cite[Lemma 2.6]{LaMM} the linear subspace $V$ consists of commutative matrices. Since $\widetilde{t}(\alpha)$ is an identity matrix, we have $Id \in V$. We also clearly have $\dim_{\mathbb{K}} V \leq n$.
    
    We can now choose a linear span $\{1, \: v_1, \: v_2, \: \ldots, \: v_{n-1}\}$ of $V$ (here $1$ is $Id$).
    We define a $\mathbb{K}[x_1, \: x_2, \: \ldots, \: x_{n-1}]$-module structure on $V$ by
    $x_i \cdot b = v_i(b)$ for $i = 1, \: 2, \: \ldots, \: n-1$ and for all $b \in B$.
    
    By commutativity it is indeed a module structure, since $x_i \cdot (x_j \cdot b) = v_i(v_j(b)) = v_j(v_i(b)) = x_j \cdot (x_i \cdot b)$. One can read more about such a construction in \cite{jelisiejew2021components}.
    Let us denote the module as $M_{\widetilde{t}}$.
    By construction its structure tensor $\sigma_{M_{\widetilde{t}}}$ is isomorphic to $\widetilde{t}$.
    Since $\widetilde{t} \simeq t$, we obtain that $t \simeq \sigma_{M_{\widetilde{t}}}$, so $t$ is isomorphic to $S_n$-module structure tensor.
    
    We now need to argue that $M_{\widetilde{t}}$ is non-semisimple.
    However, if the module $M_{\widetilde{t}}$ were semisimple, then by Lemma \ref{lemma-simple-module} it would be a direct sum of modules of rank $1$, so its structure tensor would be diagonal in some basis and thus its rank would be $n$. Since $t$ is isomorphic to this structure tensor, it would imply that $rank(t) = n$, which is contradiction with the assumption $rank(t) > n$.
\end{proof}

Let us now fix a degree $n$ and a module $M$.

To prove \ref{thm-main 1}, we will use an approach that directly generalizes the approach from \cite{blser_et_al:LIPIcs:2020:12686} - we will show that if $M$ is not semisimple, then $\sigma_M$ is $0$-subtight-unstable.

\begin{theorem}
\label{thm-main 2}
    The tensor $\sigma_M$ is $0$-subtight-unstable or $M$ is semisimple.
\end{theorem}

The proof is given later, after handful lemmas.

By $Ann(M)$ we denote the annihilator of module $M$ over $S$. Let $A = \sfrac{S}{Ann(M)}$ and let $\pi$ be a surjection from $S$ to $A$.

We use the quotient ring $A$ because it has much more convenient structure: in fact, as we show in Lemma \ref{lemma-artinian}, it is an Artinian ring, which will let us conclude a lot about its spectrum. By Hilbert's Nullstellensatz it will allow us to analyze affine functions over $\mathbb{K}$, which correspond to $S_{\leq 1}$. Artinian structure will also allow us to use the nilradical construction to obtain filtrations induced by powers of nilradical. However, we will have to struggle with one fundamental issue - we want the preimage of $rad(A)$ in $S$ to be non-trivial. In fact, we want more - preimage of $rad(A)$ has to have a non-trivial intersection with $S_{\leq 1}$.

The next three lemmas are quite classical commutative algebra arguments, but we present them to assure self-containment of the paper.

\begin{lemma}
\label{lemma-rad-power}
   There exist a natural number $r$ such that $rad(A)^r = 0$.
\end{lemma}
\begin{proof}
    By Hilbert's basis theorem $S$ is Noetherian ring and since $A$ is its quotient, it is Noetherian too.
    Hence $rad(A)$ is finitely generated. Let us denote its generators as $g_1, \: g_2 \:, \ldots, \: g_k$. Let $p_i$ be such positive integers that $g_i^{p_i} = 0$ for $i = 1, \: 2, \: \ldots \: k$. Let $r = k\max(p_i)$. Then $rad(A)^r = 0$.
    
    Indeed, if $a \in rad(A)$, then $a = \sum_{i=1}^k{c_i g_i}$ for some $c_i \in A$ and $a^r = (\sum_{i=1}^k{c_i g_i})^r = \sum_{a_1 + a_2 + \ldots + a_k = r}{b_{(a_1, \: a_2, \: \ldots, \: a_k)} g_1^{a_1} g_2^{a_2} \ldots g_k^{a_k}}$. For every component of the sum by the pigeonhole principle for some $i$ we have $a_i \geq \max(p_i)$, thus $g_i^{a_i} = 0$ and so $b_{(a_1, \: a_2, \: \ldots, \: a_k)} g_1^{a_1} g_2^{a_2} \ldots g_k^{a_k} = 0$.
\end{proof}

\begin{lemma}
\label{lemma-artinian}
   $A$ is an Artinian ring.
\end{lemma}
\begin{proof}
    Since $A$ is an $\mathbb{K}$-algebra, it suffices to show that $A$ has finite dimension over $\mathbb{K}$ (every descending sequence of ideals is a descending sequence of $\mathbb{K}$-linear subspaces so if $A$ has finite dimension over $\mathbb{K}$ then such a sequence clearly stabilizes).
    
    Let us define $\psi \colon A \rightarrow Hom_\mathbb{K}(M, \: M)$ by $\psi(s + Ann(M)) = (m \rightarrow sm)$.
    By definition of $Ann(M)$ the map $\psi$ is well-defined. Clearly $\psi$ is $\mathbb{K}$-linear.
    
    Let us now observe that $\psi$ is injective. Indeed, if $\psi(s + Ann(M)) = 0$, then for all $m \in M$ $sm = 0$ and thus $s \in Ann(M)$. Thus $A$ is isomorphic as a $\mathbb{K}$-linear subspace with $\im(\psi)$, which has finite dimension over $\mathbb{K}$ as a linear subspace of $Hom_\mathbb{K}(M, \: M)$ which has finite dimension.
\end{proof}

\begin{lemma}
\label{lemma-artinian-spectrum}
   The spectrum of an Artinian ring is finite and equals its maximal spectrum.
\end{lemma}
This lemma is well-known, but we add a proof for completeness.
\begin{proof}
    Let $R$ be an Artinian ring and let $p$ be its prime ideal. We will first argue that $p$ is maximal.
    Let $\pi_p$ be a projection from $R$ to $\sfrac{R}{p}$. Let us observe that $\sfrac{R}{p}$ is also Artinian - for every descending sequence of ideals $I_1, \: I_2, \: \ldots$ in $\sfrac{R}{p}$ the sequence of preimages $\pi_p^{-1}(I_1), \: \pi_p^{-1}(I_2), \: \ldots$ in $R$ stabilizes and thus the sequence $I_1, \: I_2, \: \ldots$ also stabilizes.
    The ring $\sfrac{R}{p}$ is a domain and it suffices to show that it is a field.
    Let $x \in \sfrac{R}{p}, \: x \neq 0$. Let us consider the descending sequence $(x) \supseteq (x^2) \supseteq \ldots$ Since it stabilizes, for some $k$ holds equality $(x^k) = (x^{k+1})$ with $k \geq 1$. Thus there exist $a \in \sfrac{R}{p}$ such that $x^k = x^{k+1}a$. So $x^k(1-xa) = 0$ and since $\sfrac{R}{p}$ is a domain and $x \neq 0$ we obtain $xa = 1$, so $x$ is invertible. So $\sfrac{R}{p}$ is a field and thus $p$ is a maximal ideal of $R$.
    
    We will now show that maximal spectrum of $R$ is finite.
    Let us suppose otherwise. Then there exists an infinite sequence of maximal ideals $\mathfrak{m_i}$. Let us consider the descending sequence $\mathfrak{m_1}, \: \mathfrak{m_1} \cap \mathfrak{m_2}, \: \mathfrak{m_1} \cap \mathfrak{m_2} \cap \mathfrak{m_3}, \: \ldots$
    Since it stabilizes there exists such $l$ that $\mathfrak{m_1} \cap \mathfrak{m_2} \cap \ldots \cap \mathfrak{m_l} = \mathfrak{m_1} \cap \mathfrak{m_2} \cap \ldots \cap \mathfrak{m_l} \cap \mathfrak{m_{l+1}}$. Clearly $\mathfrak{m_i} \not \subseteq \mathfrak{m_j}$ for $i \neq j$.
    Let $a_i \in \mathfrak{m_i} \setminus \mathfrak{m_{l+1}}$ for $i = 1, \: 2, \: \ldots, \: l$.
    Then $a_1 a_2 \ldots a_l \in \mathfrak{m_1} \cap \mathfrak{m_2} \cap \ldots \cap \mathfrak{m_l}$, so by our assumption $a_1 a_2 \ldots a_l \in \mathfrak{m_1} \cap \mathfrak{m_2} \cap \ldots \cap \mathfrak{m_l} \cap \mathfrak{m_{l+1}}$, so $a_1 a_2 \ldots a_l \in \mathfrak{m_{l+1}}$.
    Thus by primeness of $\mathfrak{m_{l+1}}$ there exists $c$ such that $a_c \in \mathfrak{m_{l+1}}$. It is a contradiction with the definition with $a_i$.
\end{proof}

We now will be using definitions referring to block tensors, which are introduced in Section \ref{section-definitions}.

\begin{proposition}
\label{proposition-subtight}
   If $\pi^{-1}(rad(A)) \cap S_{\leq 1} \neq 0$, then $\sigma_M$ is $0$-subtight-unstable.
\end{proposition}
\begin{proof}
    We will use the same approach as in {\cite[Example 14]{blser_et_al:LIPIcs:2020:12686}}.
    
    By $r$ we denote the minimal natural number such that $rad(A)^r = 0$; by Lemma \ref{lemma-rad-power} such a natural number exists.
    Let us take the sequence
    \[M \supseteq rad(A) M \supseteq rad(A)^2 M \supseteq \ldots \supseteq rad(A)^r M = 0\].
    Let us take the sequence $(M_k)$
    of linear subspaces satisfying:
    
    \begin{itemize}
        \item $rad(A)^{k-1} M = M_{k-1} \oplus rad(A)^{k} M$ for $r > k > 1$,
        \item $M_r = 0$.
    \end{itemize}
    
    By definition $(M_i)_{i=0} ^{r}$ is a decomposition of $M$.
    
    Let $(R_i)_{i=0} ^{r}$ be a decomposition of $S_{\leq 1}$ induced by $\pi^{-1}({rad (A)}^k) \cap S_{\leq 1}$, i.e. $\pi^{-1}({rad (A)}^k) \cap S_{\leq 1} = R_k \oplus \pi^{-1}({rad (A)} ^{k+1}) \cap S_{\leq 1}$, $R_r = 0$.
    
    We have $R_1 \oplus R_2 \oplus \ldots \oplus R_r = \pi^{-1}(rad(A)) \cap S_{\leq 1}$, so the assumption $R_1 \oplus R_2 \oplus \ldots \oplus R_r \neq 0$.
    
    For all $i, \: j \geq 0$ we have
    $R_i \: M_j \subseteq (rad (A)^i + Ann(M)) \: rad (A)^j M = rad (A)^{i+j} M = \bigoplus_{k \geq i+j} M_k$.
    
    Let us now consider the numbering $a_1(i) = a_2(i) = i. \: a_3(i) = -i$.
    The structure tensor $\sigma_M$ is a block tensor with decompositions obtained from decompositions of $R_i$ and $M_i$.
    Let $I_1, \: I_2, \: I_3$ be its indexing sets.
    As we observed above, for all $(i_1, \: i_2, \: i_3)$ such that $a_1(i_1) + a_2(i_2) + a_3(i_3) > 0$ it holds that $\sigma_M(R_{i_1}, \: M_{i_2}, \: M_{i_3}) = 0$.
    It means that $\sigma_M$ has $0$-subtight support. Let $(f_1, \: f_2, \: f_3)$ be a relative block format of $\sigma_M$. Let $\overline{a}_k = \sum_{i \in I_k} f_k(i) a_k(i)$ for $k = 1, \: 2, \: 3$.
    We need to to show that $\overline{a}_1 + \overline{a}_2 + \overline{a}_3 > 0$.
    Since decompositions on the second and the third coordinate are dual, we have $I_2 = I_3, \: f_2 = f_3$ and thus
    $\overline{a}_2 + \overline{a}_3 = \sum_{i \in I_2} f_2(i) \: a_2(i) + \sum_{i \in I_3} f_3(i) \: a_3(i) = \sum_{i \in I_2} f_2(i) \: i + \sum_{i \in I_3} f_3(i) \: (-i) = 0$.
    So we just need to show $\overline{a}_1 > 0$.
    We have $\overline{a}_1 = \sum_{i \in I_1} f_1(i) \: a_1(i) = \sum_{i \in I_1} f_1(i) \: i$.
    In our case $I_1 = \{0, \: 1, \: \ldots, \: r\}$ and as we observed before $R_1 \oplus R_2 \oplus \ldots \oplus R_r \neq 0$, so there exist $s \in \{1, \: 2, \: \ldots, \: r\}$ such that $f_1(s) > 0$.
    Thus we have $\sum_{i \in I_1} f_1(i) \: i = \sum_{i = 0}^r f_1(i) \: i \geq f_1(s) \: s > 0$, which ends the proof.
\end{proof}

Now we will argue that $\pi^{-1}(rad(A)) \cap S_{\leq 1}$ is zero only for semisimple modules.
To prove it, we will analyze the spectrum of $A$ and the support of $M$. 

Let us define $supp \: M$ as a set of such maximal ideals $\mathfrak{m}$ of $S$ that $\mathfrak{m}M \neq M$.
Let us also denote by $V(I)$ the Zariski closure of ideal $I$, i.e. the set of ideals which include $I$.

\begin{lemma}
\label{lemma-supp-v}
   It holds that $supp \: M = V(Ann(M))$.
\end{lemma}
\begin{proof}
    First we argue that $V(Ann(M))$ contains only maximal ideals. Since by Lemma \ref{lemma-artinian} we have that $A = \sfrac{S}{Ann(M)}$ is Artinian, by Lemma \ref{lemma-artinian-spectrum} all its prime ideals are maximal.
    Taking preimage induces bijection between $\Spec(\sfrac{S}{Ann(M)})$ and $V(Ann(M))$ and this bijection preserves maximality of an ideal.
    Thus all elements of $V(Ann(M))$ are maximal ideals.
    
    Let now $\mathfrak{m}$ be a maximal ideal of $S$. We need to show that $Ann(M) \not \subseteq \mathfrak{m} \iff \mathfrak{m}M = M$.
    
    On the one hand, if $Ann(M) \not \subseteq \mathfrak{m}$, then $\mathfrak{m} + Ann(M) = (1)$ and thus \[\mathfrak{m}M = \mathfrak{m}M + Ann(M)M = (\mathfrak{m} + Ann(M))M = (1)M = M.\]
    
    On the other hand, if $\mathfrak{m}M = M$, then by Nakayama's lemma there exist an element $m_0 \in \mathfrak{m}$ such that $(1-m_0)M = 0$ and thus $1-m_0 \in Ann(M)$. Since $1-m_0 \not \in \mathfrak{m}$, it holds that $Ann(M) \not \subseteq \mathfrak{m}$.
\end{proof}

\begin{lemma}
\label{lemma-supp-<=-deg}
   We have the inequality $|supp \: M | \leq \deg M$ and if the equality holds, then $M$ is semisimple.
\end{lemma}
\begin{proof}
    By Lemma \ref{lemma-supp-v} the set $supp \: M$ can be treated as a subset of $SpecMax(A)$.
    By Lemma \ref{lemma-artinian} we have that $A$ is Artinian and thus by Lemma \ref{lemma-artinian-spectrum} we have that $SpecMax(A)$ is finite, so $supp \: M$ is also finite.
    
    Let $supp \: M = \{\mathfrak{m}_1, \: \mathfrak{m}_2, \: \ldots, \: \mathfrak{m}_k \}$.
    
    We have $\deg M = dim_\mathbb{K} M \geq dim_\mathbb{K} \sfrac{M}{\mathfrak{m}_1M \cap \mathfrak{m}_2M \cap \ldots \cap \mathfrak{m}_kM}$.
    
    By the Chinese Remainder Theorem we have \[ \sfrac{M}{\mathfrak{m}_1M \cap \mathfrak{m}_2M \cap \ldots \cap \mathfrak{m}_kM} \simeq \sfrac{M}{\mathfrak{m}_1M} \times \sfrac{M}{\mathfrak{m}_2M} \times \ldots \times \sfrac{M}{\mathfrak{m}_kM}.\]
    By definition of $\mathfrak{m}_i$ we have $\sfrac{M}{\mathfrak{m}_iM} \neq 0$ and thus $dim_\mathbb{K} (\sfrac{M}{\mathfrak{m}_iM}) \geq 1$ and so \[ dim_\mathbb{K} (\sfrac{M}{\mathfrak{m}_1M} \times \sfrac{M}{\mathfrak{m}_2M} \times \ldots \times \sfrac{M}{\mathfrak{m}_kM}) \geq k .\]
    
    Thus $\deg M \geq \deg \sfrac{M}{\mathfrak{m}_1M \cap \mathfrak{m}_2M \cap \ldots \cap \mathfrak{m}_kM} \geq k = | supp \: M |$, which proves the first part of the statement.
    
    Let us now suppose that equality holds, i.e. $\deg M = k$.
    Thus all inequalities from the first part must be equalities and hence $\mathfrak{m}_1M \cap \mathfrak{m}_2M \cap \ldots \cap \mathfrak{m}_kM = 0$ and for all $i$ we have $dim_\mathbb{K} \sfrac{M}{\mathfrak{m}_iM} = 1$.
    First equality and Chinese Remainder Theorem give us \[M \simeq \sfrac{M}{\mathfrak{m}_1M} \times \sfrac{M}{\mathfrak{m}_2M} \times \ldots \times \sfrac{M}{\mathfrak{m}_kM}.\]
    Now by Lemma \ref{lemma-simple-module} it suffices to show that for all $i$ we have $\sfrac{M}{\mathfrak{m}_iM} \simeq \sfrac{S}{\mathfrak{m}_i}$ as $S$-modules.
    
    Since $\sfrac{M}{\mathfrak{m}_iM}$ is a linear space of dimension one, there exist $m \in \sfrac{M}{\mathfrak{m}_iM}$ such that $\mathbb{K}m = \sfrac{M}{\mathfrak{m}_iM}$.
    
    Let us define $S$-module homomorphism $\varphi \colon S \rightarrow \sfrac{M}{\mathfrak{m}_iM}$ by $\varphi(1) = m$, $\varphi(s) = sm$.
    Since $\mathbb{K}m = \sfrac{M}{\mathfrak{m}_iM}$, the map $\varphi$ is surjective.
    Thus $\sfrac{M}{\mathfrak{m}_iM} \simeq \sfrac{S}{\ker(\varphi)}$.
    Clearly $\mathfrak{m}_i \subseteq \ker(\varphi)$ and thus $\sfrac{S}{\ker(\varphi)} \subseteq \sfrac{S}{\mathfrak{m}_i}$. However $dim_\mathbb{K} \sfrac{M}{\mathfrak{m}_iM} = 1 = dim_\mathbb{K} \sfrac{S}{\mathfrak{m}_i}$ and so the inclusion must be an equality and thus indeed $\sfrac{M}{\mathfrak{m}_iM} \simeq \sfrac{S}{\mathfrak{m}_i}$ as $S$-modules.
\end{proof}

\begin{proposition}
\label{proposition-supp-=-deg}
   If it holds that $\pi^{-1}(rad(A)) \cap S_{\leq 1} = 0$, then we have an equality $|supp \: M | = \deg M$.
\end{proposition}
\begin{proof}
    By definition $\pi^{-1}(rad(A)) = \sqrt{Ann(M)}$. The ideal $rad(A)$ is a nilradical of $A$ so it is an intersection of its all prime ideals.
    Since by Lemma \ref{lemma-artinian} we have that $A$ is Artinian, by Lemma \ref{lemma-artinian-spectrum} all its prime ideals are maximal and the number of maximal ideals is finite.
    Thus $rad(A)$ is an intersection of finite number of maximal ideals and so $\pi^{-1}(rad(A))$ is also intersection of finite number of maximal ideals.
    
    By Lemma \ref{lemma-supp-<=-deg} we have that $|supp \: M | \leq \deg M$. Let us suppose by contradiction that $|supp \: M | < \deg M$.
    By Lemma \ref{lemma-supp-v} we have that $|supp \: M | = |V(Ann(M))|$.
    Let $V(Ann(M)) = \{\mathfrak{m}_1, \: \mathfrak{m}_2, \ldots, \: \mathfrak{m}_r\}$, let us remark that by Lemma \ref{lemma-supp-v} ideals $\mathfrak{m}_i$ are maximal for all $i$.
    Moreover \[ \pi^{-1}(rad(A)) = \sqrt{Ann(M)} = \bigcap_{p_i \in V(Ann(M))} {p_i} = \mathfrak{m}_1 \cap \mathfrak{m}_2 \cap \ldots \cap \mathfrak{m}_r .\]
    
    By Hilbert's Nullstellensatz maximal ideals in polynomial ring over algebraically closed field can be treated as points, we will be using this equivalence.
    In it, the elements of $\pi^{-1}(rad(A)) \cap S_{\leq 1}$ are exactly affine functions $f$ such that $f(\mathfrak{m}_i) = 0$ for all $i$.
    
    We have $r = |V(Ann(M))| = |supp \: M | < \deg M = \deg S_{\leq 1}$.
    So the space of affine functions has the dimension at least $r+1$.
    Thus the dimension of a subspace of affine functions such that $f(\mathfrak{m}_i) = 0$ for all $i$ is at least one.
    Thus there exist a non-zero function $f$ satisfying these conditions.
    Since $f$ is an element of $\pi^{-1}(rad(A)) \cap S_{\leq 1}$, we obtain $\pi^{-1}(rad(A)) \cap S_{\leq 1} \neq 0$, which gives us expected contradiction.
\end{proof}

\begin{proof}[Proof of Theorem \ref{thm-main 2}]
\label{proof-thm-main 2}
    Let us suppose that $M$ is not semisimple. We need to prove that $\sigma_M$ is $0$-subtight-unstable.
    
    By Lemma \ref{lemma-supp-<=-deg} we have that $|supp \: M | < \deg M$. Thus by Proposition \ref{proposition-supp-=-deg} we have $\pi^{-1}(rad(A)) \cap S_{\leq 1} \neq 0$. Proposition \ref{proposition-subtight} implies that $\sigma_M$ is $0$-subtight-unstable.
\end{proof}

\begin{theorem}
\label{thm-slice-rank-bound}
    Let $t \in \mathbb{K}^n \otimes \mathbb{K}^n \otimes \mathbb{K}^n$ be a non-zero $s$-subtight-unstable tensor with numbering $(a_1, \: a_2, \: a_3)$ and relative block format $(f_1, \: f_2, \: f_3)$. Let also $\overline{a}_k = \sum_{i \in I_k} f_k(i) a_k(i)$ for $k = 1, \: 2, \: 3$. Then it holds that
    \[\widetilde{SR}(t) \leq n \exp \left( -\frac{(\overline{a}_1 + \overline{a}_2 + \overline{a}_3 - s)^2}{6 \sum_{k=1}^3 \sum_{i_k \in I_k} a_k^2} \right). \]
\end{theorem}

\begin{proof}
    The proof is given in the proof of \cite[Theorem 22]{blser_et_al:LIPIcs:2020:12686}. Although in a statement there is an assumption that tensors are over $\mathbb{C}$, there proof does not use this assumption (it operates on supports of the tensors), which is also noted in \cite[Remark 26]{blser_et_al:LIPIcs:2020:12686}.
    
    Note that the fraction is well defined, since if it held that $a_k \equiv 0$ for $k = 1, \: 2, \: 3$, then $\overline{a}_k = 0$. We have that $\overline{a}_1 + \overline{a}_2 + \overline{a}_3 > s$, so $s < 0$. However since $t$ is non-zero, its support is non-empty, and by definition for some $(i_1, \: i_2, \: i_3)$ it holds that $0 = a_1(i_1) + a_2(i_2) + a_3(i_3) \leq s$. So we would have $s \geq 0 > s$, which is a contradiction.
\end{proof}

Having Theorem \ref{thm-main 2} and Theorem \ref{thm-slice-rank-bound} proved, we are almost ready to prove the main theorem, using similar tools as Bl\"aser and Lysikov. However we have to slightly change the assumptions, so we will need to prove one more lemma.

Let us remind that for tensor $t \in V_1 \otimes V_2 \otimes V_3$ we denote by $N(t)$ the maximal rank of flattening of $t$ (see Section \ref{section-definitions}).

\begin{lemma}
\label{lemma-N}
   We have an equality $N(\sigma_M) = n$.
\end{lemma}

\begin{proof}
    Clearly $N(\sigma_M) \leq n$, since $n = \dim_{\mathbb{K}} S_{\leq 1} = \dim_{\mathbb{K}} M$ and $\sigma_M \in S_{\leq 1}^* \otimes M^* \otimes M$.
    Since $1 \in S_{\leq 1}$, the matrix of $\sigma_M(1, -)$ is an identity matrix and thus the ranks of flattenings of $\sigma_M$ over the second and the third coordinate are at least $\dim_{\mathbb{K}} M = n$, so $N(\sigma_M) \geq n$, and thus the equality must hold.
\end{proof}

We now present more general result which implies Theorem \ref{thm-main 1}. In particular, to prove Theorem \ref{thm-main 1} we will use $0$-subtight-instability.

\begin{proposition}
\label{proposition-general-subtight}
   In $\mathbb{K}^{n \times n \times n}$ with $n$ fixed, it is impossible to prove $\omega = 2$ using arbitrary restrictions from powers of elements of tensor family $G_n$ such that for every $t \in G_n$ tensor $t$ is $s$-subtight-unstable and satisfies $N(t) = n$.
\end{proposition}
\begin{proof}
    We use the approach as in {\cite[Theorem 16]{blser_et_al:LIPIcs:2020:12686}}, but we have slightly weaker assumptions - we do not require conciseness.
    
    Note that zero tensor does not belong to $G_n$, since $N(\mathbf{0}) = 0$.
    
    If we let
    \[ B(n) \coloneqq \inf \left \{ \frac{\log N(t)}{\log \widetilde{SR} (t)} \mid t \in G_n \right \}, \]
    then by {\cite[Proposition 15]{blser_et_al:LIPIcs:2020:12686}} the irreversibility of any tensor with $N(t) \leq n$ is bounded from below by $B(n)$, and by \cite[Theorem 9]{Christandl_Vrana_Zuiddam2} {\cite[Theorem 7]{blser_et_al:LIPIcs:2020:12686}} the best bound on $\omega$ we can get is at least $2 B(n)$.
    
    By the assumption we have for all $t \in G_n$ that $N(t) = n$. Also by the assumption all tensors from family $G_n$ are $s$-subtight-unstable, so by Theorem \ref{thm-slice-rank-bound} we have for all $t \in G_n$ the bound
    \[ \widetilde{SR}(t) \leq n \exp \left( -\frac{(\overline{a}_1(t) + \overline{a}_2(t) + \overline{a}_3(t) - s)^2}{6 \sum_{k=1}^3 \sum_{i_k \in I_k} a_k(t)^2} \right). \]
    
    Let us denote
    \[ \overline{B}(n)(t) = \frac{\log n}{\log n -\frac{(\overline{a}_1(t) + \overline{a}_2(t) + \overline{a}_3(t) - s)^2}{6 \sum_{k=1}^3 \sum_{i_k \in I_k} a_k(t)^2} } \]
    and
    \[ \overline{B}(n) = \inf \left \{ \overline{B}(n)(t) \mid t \in G_n \right \}. \]
    
    We can bound $B(n)$ from below by $\overline{B}(n)$. So it is sufficient to show that $\overline{B}(n) > 1$.
    
    Clearly the term $\frac{(\overline{a}_1(t) + \overline{a}_2(t) + \overline{a}_3(t) - s)^2}{6 \sum_{k=1}^3 \sum_{i_k \in I_k} a_k(t)^2}$ is non-negative. Moreover, all tensors from the family $G_n$ are $s$-subtight-unstable, so by definition the term $\overline{a}_1(t) + \overline{a}_2(t) + \overline{a}_3(t) - s$ is strictly positive. It implies that for $t \in G_n$ the value of $\overline{B}(n)(t)$ is strictly greater than one. Moreover, the value of $\overline{B}(n)(t)$ depends only on the support of the tensor $t$. Since for fixed $n$ the number of possible supports of tensors is finite, the set of possible values of $\overline{B}(n)(t)$ is also finite.
    Thus we have
    \[ \overline{B}(n) = \inf \left \{ \overline{B}(n)(t) \mid t \in G_n \right \} = \min \left \{ \overline{B}(n)(t) \mid t \in G_n \right \} > 1, \]
    which ends the proof.
\end{proof}

\begin{proof}[Proof of Theorem \ref{thm-main 1}]
\label{proof-thm-main-1}
    By Theorem \ref{thm-main 2} all tensors $t \in F_n$ are $0$-subtight-unstable.
    By Lemma \ref{lemma-N} it also holds that $N(t) = n$.
    Applying result of Proposition \ref{proposition-general-subtight} for $G_n = F_n$, we obtain the claim.
\end{proof}

Now we present the exact bound of irreversibility of the structure tensor of a non-semisimple module. Note that we have already proved that this irreversibility is strictly greater than one.

\begin{corollary}
\label{corollary_exact_bound}
   If the module $M$ is non-semisimple, then the irreversibility of its structure tensor $\sigma_M$ is at least
   \[ \left( 1 - \left( \frac{ \left( \sum_{i=0}^{r-1}i s_i \right)^2 }{6 n^2 \log n \left( 2 \sum_{i=0}^{r-1} m_i^2 + \sum_{i=0}^{r-1} s_i  \right) } \right) \right)^{-1} ,\]
   where $r$ is minimal such that $rad(A)^r = 0$ ($r$ exists by Lemma \ref{lemma-rad-power}), $t_i = \dim_{\mathbb{K}} \pi^{-1}({rad (A)}^i) \cap S_{\leq 1}$ $s_i = t_i - t_{i+1}$, $m_i = \dim_{\mathbb{K}} \left( {rad (A)}^{i}M / {rad (A)}^{i+1}M \right) $.
\end{corollary}
\begin{proof}
    As we observed, if $M$ is non-semisimple, then its structure tensor is $0$-subtight unstable and satisfies $N(\sigma_M) = n$, so it satisfies the assumptions from Proposition \ref{proposition-general-subtight}, so we can use the bound
    \[ i(\sigma_M) \geq \frac{\log n}{\log n -\frac{(\overline{a}_1(t) + \overline{a}_2(t) + \overline{a}_3(t))^2}{6 \sum_{k=1}^3 \sum_{i_k \in I_k} a_k(t)^2} }. \]
    Now it is sufficient to use definitions of $a_i$ and $\overline{a}_i$ from the proof of Proposition \ref{proposition-subtight} and we obtain the desired bound. Analogical bound for the structure tensors of non-semisimple algebras is given in \cite[Corollary 23]{blser_et_al:LIPIcs:2020:12686}.
\end{proof}

\section{Characteristic zero}
\label{section-generalisations}

We will now show that in Proposition \ref{proposition-general-subtight} we can weaken the assumption of $s$-subtight-instability for fields of characteristic zero.

\begin{proposition}
\label{proposition-general-comb_for_c}
   In $\mathbb{C}^{n \times n \times n}$ with $n$ fixed, it is impossible to prove $\omega = 2$ using arbitrary restrictions from powers of elements of tensor family $G_n$ such that for every $t \in G_n$ tensor $t$ is unstable and satisfies $N(t) = n$.
\end{proposition}
\begin{proof}
    From {\cite[Theorem 5]{blser_et_al:LIPIcs:2020:12686}} (originally in \cite{Brion}) and {\cite[Theorem 4]{blser_et_al:LIPIcs:2020:12686}} (originally in \cite{Christiandl_slice_rank}) it follows that the set of possible values of $\widetilde{SR} (t)$ for tensors in $\mathbb{C}^{n \times n \times n}$ is finite.
    Therefore, the set of possible ratios $\frac{\log N(t)}{\log \widetilde{SR} (t)}$ is also finite.
    So we can let
    \[ B(n) \coloneqq \min \left \{ \frac{\log N(t)}{\log \widetilde{SR} (t)} \mid t \in G_n \right \}, \]
    and by {\cite[Proposition 15]{blser_et_al:LIPIcs:2020:12686}} the irreversibility of any tensor with $N(t) \leq n$ is bounded from below by $B(n)$, and also by \cite[Theorem 9]{Christandl_Vrana_Zuiddam2} {\cite[Theorem 7]{blser_et_al:LIPIcs:2020:12686}} the best bound on $\omega$ we can get is at least $2 B(n)$.
    
    Since all tensors form $G_n$ are unstable, we have by {\cite[Theorem 3]{blser_et_al:LIPIcs:2020:12686}} that for all $t \in G_n$ we have $\widetilde{SR} (t) < n$.
    Since we assumed that for all $t \in G_n$ we have $N(t) = n$, it holds that $2 B(n) > 2$, so it is impossible to show $\omega = 2$.
\end{proof}

\begin{proposition}
\label{proposition-general-comb}
   In $\mathbb{K}^{n \times n \times n}$ with $n$ fixed with characteristic of $\mathbb{K}$ equal to zero, it is impossible to prove $\omega = 2$ using arbitrary restrictions from powers of elements of tensor family $G_n$ such that for every $t \in G_n$ tensor $t$ is combinatorially unstable and satisfies $N(t) = n$.
\end{proposition}
\begin{proof}
    Let us suppose that it is possible to prove $\omega = 2$ over $\mathbb{K}$ using arbitrary restrictions. We will prove that then it is possible to prove $\omega = 2$ over $\mathbb{C}$ using arbitrary restrictions and obtain contradiction with Proposition \ref{proposition-general-comb_for_c}.
    
    Possibility to prove $\omega = 2$ over $\mathbb{K}$ using arbitrary restrictions means that there exists a sequence of pairs $(t_i, \: N_i)$ such that for all $i$ it holds that $t_i \in G_n$ and using arbitrary restrictions (which are basically some linear equations) we can obtain from $t_i^{\otimes N_i}$ large matrix multiplication, i.e. such that for a limit at infinity we obtain $\omega = 2$.
    
    For a tensor $t_i$ let $(x_{i_j})$ be a (finite) tuple of vectors corresponding to the flatenning of $t_i$ with maximal rank (each of three flattenings is a linear map, we choose basis of tensors $x_{i_j}$ for which contractions with $t_i$ give linear space of maximal dimension). Let also $Z_i$ be a set of coefficients of arbitrary restrictions (which are basically linear equations) used with $t_i$. All sets $Z_i$ are finite.
    
    Now we want to transfer this sequence to $\mathbb{C}$. Firstly let us observe that we only need to consider powers of tensors from sequence $t_i$, which means that if $\mathbb{K}$ is large we can drop some part of it. More precisely, let $\mathbb{\widetilde{L}}$ be a smallest subfield of $\mathbb{K}$ such that $\mathbb{\widetilde{L}}$ contains all elements of sets $Z_i$ and tensors $t_i$ and $x_{i_j}$ (which means that we treat tensors as tuples and want  $\mathbb{\widetilde{L}}$ to contain all elements of corresponding tuples).
    Let the field $\mathbb{L}$ be an algebraic closure of the field $\mathbb{\widetilde{L}}$. Since $\mathbb{K}$ is algebraically closed, $\mathbb{L}$ is a subfield of $\mathbb{K}$. The field $\mathbb{\widetilde{L}}$ is a countable field and thus $\mathbb{L}$ is also a countable field as an algebraic closure of a countable field, so $\mathbb{L}$ is a countable extension of $\mathbb{Q}$.
    By construction of $\mathbb{L}$ there exists an injection from $\mathbb{L}$ to $\mathbb{K}$. This injection is a bijection on tensors $t_i$, so it preserves the whole construction of obtaining large matrix multiplications using arbitrary restrictions.
    Since the field $\mathbb{L}$ is a countable extension of $\mathbb{Q}$, there exists an injection from $\mathbb{L}$ to $\mathbb{C}$.
    Let us denote by $t_i^\prime$ the preimages of $t_i$ in injection from $\mathbb{L}$ to $\mathbb{K}$ and by $\widetilde{t_i}$ the images of $t_i^\prime$ in injection from $\mathbb{L}$ to $\mathbb{C}$. Now we can induce a construction of large matrix multiplications in $\mathbb{C}$ using tensors $\widetilde{t_i}$ and arbitrary restrictions.
    
    We now want to obtain a contradiction with Proposition \ref{proposition-general-comb_for_c}. To do so, we need to prove that tensors $\widetilde{t_i}$ are unstable and $N(\widetilde{t_i}) = n$.
    By {\cite[Theorem 20]{blser_et_al:LIPIcs:2020:12686}} combinatorial instability implies instability, so it is sufficient to show that tensors $\widetilde{t_i}$ are unstable and $N(\widetilde{t_i}) = n$.
    First we show that $t_i^\prime$ are combinatorially unstable and $N(t_i^\prime) = n$.
    The latter statement is quite clear, since we have that $N(t_i) = n$, which is obtained by contractions with tensors $x_{i_j}$ (by definition) and because tensors $x_{i_j}$ can be injected to $\mathbb{L}$, we have also $N(t_i^\prime) = n$.
    Being combinatorially unstable is also quite clear - because we added to $\mathbb{L}$ tensors $t_i$, supports of these tensors over $\mathbb{K}$ and $\mathbb{L}$ are the same sets. Since being combinatorially unstable is a property of tensor support, and supports over both fields are equal, being combinatorially unstable over $\mathbb{K}$ implies being combinatorially unstable over $\mathbb{L}$.
    Because $\mathbb{L}$ can be injected to $\mathbb{C}$, we can inject $x_{i_j}$ to $\mathbb{C}$, so analogically $\widetilde{t_i}$ are combinatorially unstable. Also $t_i$ can be injected to $\mathbb{C}$, so the support of $\widetilde{t_i}$ is equal to supports of $t_i$ and $t_i^\prime$, so we also have $N(\widetilde{t_i}) = n$.
\end{proof}

\begin{remark}
\label{remark-characteristic}
It is well-known (a result due to Sch\"onhage) that $\omega$ depends only on the characteristic of the field, not on the field itself \cite[Corollary 15.18]{BurgisserBook}. However, in Proposition \ref{proposition-general-comb_for_c} and Proposition \ref{proposition-general-comb} we do not show that $\omega$ in characteristic zero has some specific value, but that it cannot be proved that $\omega = 2$ using arbitrary restrictions from the class of starting tensors. Thus, the generalisation from the case of $\mathbb{C}$ to the case of any field of characteristic zero needed justification.
\end{remark}

\printbibliography
\end{document}